\documentclass[aps,prl]{revtex4}
\usepackage{amsfonts,amssymb,amsmath}            
\usepackage{lmodern}
\usepackage[]{graphics,graphicx}            
\usepackage{amsthm}
\usepackage{bbold,enumerate,mathtools}
    \usepackage{hyperref}

\newtheorem{problem}{Problem}
\newtheorem{definition}{Definition}
\newtheorem{lemma}{Lemma}

\newcommand{\ket}[1]{\left | #1 \right\rangle}
\newcommand{\bra}[1]{\left \langle #1 \right |}
\newcommand{\half}{\mbox{$\textstyle \frac{1}{2}$}}

\newcommand{\Tr}{\text{Tr}}
\newcommand{\braket}[2]{\left\langle #1|#2\right\rangle}
\newcommand{\proj}[1]{\ket{#1}\bra{#1}}
\newcommand{\identity}{\mathbb{1}}
\renewcommand{\epsilon}{\varepsilon}
\begin{document}

\title{Generating Quantum States through Spin Chain Dynamics}
\date{\today}

\author{Alastair \surname{Kay}}
\affiliation{Department of Mathematics, Royal Holloway University of London, Egham, Surrey, TW20 0EX, UK}
\email{alastair.kay@rhul.ac.uk}
\begin{abstract}
Spin chains can realise perfect quantum state transfer between the two ends via judicious choice of coupling strengths. In this paper, we study what other states can be created by engineering a spin chain. We conclude that, up to local phases, all single excitation quantum states with support on every site of the chain can be created. We pay particular attention to the generation of W-states that are superposed over every site of the chain.
\end{abstract}

\maketitle

\section{Introduction}

Spin chains are good models for a large variety of one-dimensional systems that exhibit quantum effects. For the past decade, these systems have been intensively studied from the perspective of quantum information -- understanding how these chains can be used to implement the tasks that we specify. Perfect state transfer (see, for example, \cite{bose2003,christandl2004,burgarth2005,kay2006,kay2010-a}) -- making particular choices of the couplings strengths and magnetic fields such that a single qubit state $\ket{\psi}$ on the first spin at time $t=0$ arrives perfectly at the last spin at the state transfer time, $t_0$ -- is the typical case examined. The same solutions generate entanglement, both bipartite \cite{christandl2005} and that required for cluster states \cite{clark2005}. Simple modification of these coupling schemes permits fractional revivals \cite{dai2010,kay2010-a,banchi2015,genest2016} -- superposing the input state over the two extremal sites of the chain. Meanwhile, modification of the form of the Hamiltonian has demonstrated that other tasks can be achieved, such as the generation of a GHZ state \cite{kay2007}.

In this paper, we address the question of what other functions a spin chain can realise. In answer, we demonstrate that a wide range of one-excitation states can be generated by evolving an excitation initially located on a single site, including the important case of the $W$ state of $N$ qubits. The solution is related to the study of inverse eigenvalue and inverse eigenmode problems \cite{gladwell2005}. However, the variant that we require is, to our knowledge, unstudied. Although we prove that a solution to this variant cannot be guaranteed, we provide a protocol that yields strategies that are sufficient for our needs.

\subsection{Setting}

The Hamiltonian of a spin chain of length $N$ is
\begin{equation}
H=\sum_{n=1}^N\frac{B_n}{2}(\identity-Z_n)+\sum_{n=1}^{N-1}\frac{J_n}{2}(X_nX_{n+1}+Y_nY_{n+1}),	\label{eqn:ham}
\end{equation}
where $X_n$ denotes the Pauli $X$ matrix applied to site $n$ (and $\identity$ elsewhere). It is excitation preserving,
$$
\left[H,\sum_{n=1}^NZ_n\right]=0,
$$
meaning that, for example, any one-excitation state (a state of $N-1$ $\ket{0}$s and one $\ket{1}$) can only be evolved into another one-excitation state. Indeed, the Hamiltonian when restricted to the first excitation subspace is described as
$$
H_1=\sum_{n=1}^NB_n\proj{n}+\sum_{n=1}^{N-1}J_n(\ket{n}\bra{n+1}+\ket{n+1}\bra{n}),
$$
where $\ket{n}:=\ket{0}^{\otimes (n-1)}\ket{1}\ket{0}^{\otimes (N-n)}$. This is a real, symmetric, tridiagonal matrix where each of the elements can be independently specified, making it ideal for the engineering tasks that we intend to study. Moreover, via the Jordan-Wigner transformation, one can readily describe the evolution of higher excitation states in terms of the evolution of single excitation states.

We will study the following problem:
\begin{problem} \label{prob:main}
Given a normalised one excitation state
$$
\ket{\alpha}=\sum_{n=1}^N\alpha_n\ket{n},
$$
find coupling strengths $\{J_n\}$, magnetic fields $\{B_n\}$, and an initial site $k$ such that there exists a time $t_0$ for which
$$
e^{-iHt_0}\ket{k}=\ket{\alpha}.
$$
\end{problem}
We do this by showing it is sufficient to ensure that the Hamiltonian $H_1$ has eigenvalues which satisfy a particular property, and by fixing one of the eigenvectors. We show that, under this particular mapping of the problem, there are instances where there is no solution, and instances when the solution is not unique. However, in the practical sense of answering \ref{prob:main}, we provide a technique that gives arbitrarily high quality results for all but a very small category of possible cases (specified by a property on the $\alpha_n$).

\section{The Hybrid Inverse Eigenvalue/mode Problem}

Consider the hybrid inverse eigenvalue/mode problem:
\begin{problem} \label{prob:supp}
Given a real, normalised vector
$$
\ket{\eta}=\sum_{n=1}^N\eta_n\ket{n},
$$
such that $\eta_1\eta_N\neq0$ and no two consecutive values $\eta_n$ and $\eta_{n+1}$ are both zero, and a set of distinct real numbers $\Lambda=\{\lambda_n\}_{n=1}^N$, find a real, symmetric, tridiagonal matrix $H_1$ with eigenvalues $\Lambda$ such that $H_1\ket{\eta}=\eta\ket{\eta}$ ($\eta\in\Lambda$).
\end{problem}
\noindent The constraints on the values of $\eta_n$ are necessary conditions for $\ket{\eta}$ to be an eigenvector of $H_1$ \cite{gladwell1986}\footnote{\cite{gladwell1986} specifies a further property on sign changes between $\eta_{n-1}$ and $\eta_{n+1}$ if $\eta_n=0$ because they imposed that all the $J_n$ should be negative. All that remains from that condition is that we cannot allow $\alpha_n=\alpha_{n+1}=0$.}. Similarly, tridiagonal matrices do not have degenerate eigenvalues.

Particular instances of Problem \ref{prob:supp} provide a solution to Problem \ref{prob:main} via the observation:
\begin{lemma} \label{lem:reduction}
A sufficient condition for a solution to Problem \ref{prob:main} is that there exists a $k\in[N]$ and $t_0\in\mathbb{R}^+$ for which
$$
\ket{\eta}=\frac{\ket{k}-\ket{\alpha}}{\sqrt{2(1-\alpha_k)}}
$$
has a solution to Problem \ref{prob:supp} with a spectrum in which $e^{-it_0\eta}=1$ and $e^{-it_0\lambda}=-1$ for all $\lambda\in\Lambda\setminus\eta$.
\end{lemma}
\begin{proof}
By definition, we have $e^{-iH_1t_0}=-(\identity-2\proj{\eta})$. So, by starting from a state $\ket{k}$,
$$
e^{-iH_1t_0}\ket{k}=-\ket{k}+2\braket{\eta}{k}\ket{\eta}=-\ket{\alpha}.
$$
Without loss of generality, we take $\eta=0$. Thus, all other eigenvalues must be half-integer multiples of $2\pi/t_0$.
\end{proof}

Note that this result is only sufficient -- it is certainly not necessary as it does not include the case of perfect state transfer or fractional revivals (for $N>3$) because these cases have $k=1$ and $\eta_2,\ldots,\eta_{N-1}=0$. 

To our knowledge, the construction of tridiagonal matrices with a specific spectrum and a specific eigenvector has not been studied, although the independent questions of inverse eigenvalue \cite{hochstadt1967} and inverse eigenmode \cite{gladwell1986} problems have been examined. As such, we are interested in categorising when solutions to Problem \ref{prob:supp} exist, and how to find them.

We start by making an observation about the necessary pattern of signs of the coupling strengths such that a specified eigenvector can correspond to a particular eigenvalue in the ordered sequence. Recall \cite{gladwell1986} that if all the $J_n$ are negative, the eigenvector with the $n^{th}$ largest eigenvalue has $N-n$ sign changes in its amplitudes. Thus, to ensure that a particular eigenvector $\ket{\eta}$ has the $n^{th}$ largest eigenvalue, find a diagonal matrix $D$, with $D^2=\identity$ such that $D\ket{\eta}$ has $N-n$ sign changes. Thus, if matrix $H_1$ has coupling strengths $J_n$ which are all negative, and an eigenvector $D\ket{\eta}$ which has $N-n$ sign changes, and thus has the $n^{th}$ largest eigenvalue, the matrix $DH_1D$ has the same magnetic fields, the coupling strengths are the same up to sign changes
$$
\text{sign}(J_m)=-D_mD_{m+1},
$$
and $\ket{\eta}$ is an eigenvector. Moreover, since $D$ is unitary, the transformation was isospectral, and $\ket{\eta}$ must have the $n^{th}$ largest eigenvector.

\begin{lemma}
Specifying a spectrum and a target eigenvector is insufficient to yield a unique solution.
\end{lemma}
\begin{proof}
By uniqueness, we mean choice of the values $\{J_n^2\}$ -- changing the signs of the $J_n$ is a triviality which we want to discount.
The Hamiltonian
$$
\left(\begin{array}{ccccc} -J_1 & J_1 & 0 & 0 & 0 \\
J_1 & -J_1-J_2 & J_2 & 0 & 0 \\
0 & J_2 & 0 & -J_2 & 0 \\
0 & 0 & -J_2 & J_1+J_2 & -J_1 \\
0 & 0 & 0 & -J_1 & J_1
\end{array}\right)
$$
where $J_2=-\sqrt{45}/J_1$ has spectrum $0,\pm 3, \pm 5$ and the 0-eigenvector is the $W$-state for two distinct values of $J_1^2$:
$$
J_1^2=\frac{17+3\sqrt{5}\pm\sqrt{102\sqrt{5}-206}}{4}.
$$
\end{proof}

\begin{lemma}\label{lem:nosol}
Problem \ref{prob:supp} does not always have a solution.
\end{lemma}
\begin{proof}
It suffices to find a counter example. To that end, fix $N=5$ and $\ket{\eta}=(\ket{1}+\ket{3}+\ket{5})/\sqrt{3}$ with a target spectrum of $\{0,\pm3,\pm5\}$ (note that this example is compatible with the specification of Lemma \ref{lem:reduction} with $t_0=\pi$). Requiring $H_1\ket{\eta}=0$ immediately restricts the structure to
$$
H_1=\left(\begin{array}{ccccc}
0 & J_1 & 0 & 0 & 0 \\
J_1 & B_2 & -J_1 & 0 & 0 \\
0 & -J_1 & 0 & -J_4 & 0 \\
0 & 0 & -J_4 & B_4 & J_4 \\
0 & 0 & 0 J_4 & 0
\end{array}\right)
$$
We then fix $\Tr(H_1)=B_2+B_4=0$, i.e.\ $B_4=-B_2$. Next, $\Tr(H_1^3)=0=6B_2(J_1^2-J_4^2)$. We take the two cases of $B_2=0$ and $J_1^2=J_4^2$ separately. If $B_2=0$, then we can solve $J_1^2$ and $J_4^2$ simultaneously in
\begin{eqnarray*}
\Tr(H_1^2)&=34&=4(J_1^2+J_4^2)	\\
\Tr(H_1^4)&=706&=4(2J_1^4+2J_4^4+J_1^2J_4^2)
\end{eqnarray*}
There are no non-negative solutions. Similarly, for $J_1^2=J_4^2$, one has to simultaneously solve
\begin{eqnarray*}
\Tr(H_1^2)&=34&=2(B_2^2+4J_4^2)	\\
\Tr(H_1^4)&=706&=2((B_2^2+4J_4^2)^2-6J_4^4)
\end{eqnarray*}
which, again, has no solutions.
\end{proof}

\section{Arbitrarily Accurate Solutions}

It is not possible to realise any arbitrary assignment of eigenvalues and a single eigenvector. However, Problem \ref{prob:main} does not require a specific spectrum, only that certain general properties are obeyed. So, is it still possible to select a target spectrum such that the conditions of Lemma \ref{lem:reduction} are satisfied and a solution to Problem \ref{prob:main} exists?

\begin{lemma} \label{lem:numeric}
For any target state $\ket{\alpha}$ satisfying the conditions of Lemma \ref{lem:reduction} and Problem \ref{prob:supp}, and any sufficiently small prescribed accuracy $\epsilon$, there exists a time $t_0\sim1/\epsilon$ such that $\ket{k}\mapsto\ket{\alpha}$ to accuracy $1-O(\epsilon^2)$.
\end{lemma}
\begin{proof}
We start by solving the inverse eigenmode problem \cite{gladwell1986} without any regard for the eigenvalues, fixing $\ket{\eta}$ to have 0 eigenvalue. Since
$$
\eta_{n-1}J_{n-1}+\eta_nB_n+\eta_{n+1}J_n=0\qquad\forall n, 
$$
then provided $\eta_n\neq0$, any choice of $J_n$ fixes the $B_n$. So, we just make a choice, say $J_n=1$. We call this Hamiltonian $H_\eta$. To correct the spectrum, we follow \cite{karbach2005}:
\begin{itemize}
\item Pick an accuracy parameter $\epsilon$ (smaller than half the smallest gap between eigenvalues in $H_\eta$).
\item Truncate the eigenvalues of $H_\eta$ to the nearest multiple of $\epsilon$.
\item Shift all the eigenvalues except the 0 value by $\pm\half\epsilon$. The choice of $\pm$ does not matter, and can be made in order to minimise the change in the eigenvalues, which need never be larger than $\epsilon/4$. This ensures that the ordering of the eigenvalues is maintained.
\item Take the values $\{\braket{1}{\lambda_n}\}$, where $\ket{\lambda_n}$ are the eigenvectors of $H_\eta$, and use these along with the target spectrum to calculate, via the Lanczos method \cite{gladwell2005}, a new Hamiltonian $\tilde H$.
\end{itemize}
The output, $\tilde H$, is guaranteed to have a spectrum that achieves the desired phases (up to a global phase of $-1$) in a time $t_0=2\pi/\epsilon$. A solution to this always exists \cite{gladwell2005}. While the 0 eigenvector is no longer $\ket{\eta}$, but $\ket{\eta_{\text{actual}}}$, since $\tilde H$ is only a perturbation of $H_\eta$, it should not be significantly different.

How different is it? We estimate $F=\braket{\eta}{\eta_{\text{actual}}}$ as an accuracy parameter (the overlap between the state produced and the desired state is $1-2(1-\alpha_1)(1-F)$ if the excitation is initially placed on site 1). By construction, $F$ is real since both $\ket{\eta}$ and $\ket{\eta_{\text{actual}}}$ are real. If $U$ and $\tilde U$ diagonalise $H_\eta$ and $\tilde H$ respectively, then the calculation of $F$ is equivalent to $\bra{m}U^\dagger\tilde U\ket{m}$ where $m$ is the index of the relevant eigenvector: $U\ket{m}=\ket{\eta}$. However, $U$ and $\tilde U$ must be very similar, so we choose an expansion
$$
U^\dagger\tilde U=(\identity+i\epsilon K)(\identity-i\epsilon K)^{-1},
$$
which maintains unitarity and the limit $\tilde U\rightarrow U$ as $\epsilon\rightarrow 0$, where $K$ is Hermitian \cite{downing1956}. Expanding for small $\epsilon$,
$$
F=\bra{m}\identity+2i\epsilon K-2\epsilon^2K^2+O(\epsilon^3)\ket{m}.
$$
Since $F$ is real, and the diagonal of $K$ is real, the diagonal of $K$ must be 0, such that we are left with the second order term, as required.
\end{proof}
By continuity of the spectral properties of the Hamiltonian, we infer that a perfect realisation must exist. Thus, as a special case, we can create any state with real, non-zero amplitudes on every site of the chain, including states such as the $W$ state. For example,
$$
H_1=\left(
\begin{array}{ccccc}
0.80985122 & 1.00004543 & 0 & 0 & 0 \\
1.00004543 & 0.23665936 & 1.00033274 & 0 & 0	\\
0 & 1.00033274 & -1.99911163 & 0.99971024 & 0	\\
0 & 0 & 0.99971024 & -1.9996369 & 1.00055901	\\
0 & 0 & 0 & 1.00055901 & -0.99954444
\end{array}
\right),
$$
with parameter $\epsilon=0.001687714$ evolves $\ket{1}\rightarrow\ket{\alpha}$ where $\alpha$ has an overlap with the 5-qubit W-state of 0.999999998.

In previous studies of perfect state transfer, it has been deemed acceptable for the arriving state to only be exact up to the application of a phase gate since, for the created state to be any use, one must have some local control at each of the output sites and that would be capable of compensating for the phase. Were one to make the same relaxation here, then any complex state (with non-zero amplitudes) can be realised simply by redefining $\ket{\alpha}\mapsto\sum_n|\alpha_n|\ket{n}$ first, and then applying local phases $R_Z(\text{Arg}(\alpha_n))$ on each of the sites at the time $t_0$.

\subsection{Error Scaling} In Lemma \ref{lem:numeric}, we have proven that the error term scales as $\epsilon^2\bra{m}K^2\ket{m}$, which immediately conveys the $\epsilon$ dependence, but disguises the $N$ dependence.  Following \cite{downing1956}, we can derive that $\bra{m}K^2\ket{m}=\sum_n|U_{nm}|^2G_n^2$ where $G$ is a diagonal matrix satisfying
\begin{equation}
\sum_{n}|U_{nm}|^2G_n=e_m \qquad \forall m\in[N]	\label{eqn:nasty}
\end{equation}
and $e_m$ is the difference between the $m^{th}$ largest intended and actual eigenvalues as a fraction of $\epsilon$. Consider
$$
\sum_m\bra{m}K^2\ket{m}=\sum_nG_n^2,
$$
which is $N$ times larger than the average error, and no smaller than the worst-case error. If $\ket{G}$ solves
$$
\left(\sum_{n,m}|U_{nm}|^2\ket{m}\bra{n}\right)\ket{G}=\sum_me_m\ket{m}
$$
(which must have a solution, even if $V=\sum_{n,m}|U_{nm}|^2\ket{m}\bra{n}$ is singular), then the error estimate is simply $\braket{G}{G}$. Thus, if $\zeta$ is the smallest non-zero singular value of $V$, we have
$$
\braket{G}{G}\leq\frac{1}{\zeta^2}\max e_m^2\leq\frac{1}{16\zeta}.
$$

To demonstrate that the scaling is not pathological, we study the special case in which $H_\eta$ has $J_n=1$ and $B_n=0$ for all $n$. This is particularly pertinent to the creation of a $W$ state. We have that
$$
V=\frac{2}{N+1}\sum_{n,m=1}^{(N+1)/2}\sin^2\left(\frac{\pi n m}{N+1}\right)\ket{n}\bra{m}.
$$
To find the eigenvalues, observe that for $N>5$, the states
$$
\ket{\frac{N+1}{2}},\quad\sum_n\ket{2n},\quad\sum_n\ket{2n-1}
$$
span a 3-dimensional subspace in which the Hamiltonian may be represented as
$$
\begin{array}{cc}
\frac{1}{4\sqrt{N+1}}\left(\begin{array}{ccc} 0 & 4 & 0 \\ 4 & \sqrt{N+1} & \sqrt{N-3} \\ 0 & \sqrt{N-3} & \sqrt{N+1} \end{array}\right) & N\equiv 3\text{ mod }4	\\
\frac{1}{4(N+1)}\left(\begin{array}{ccc} 8 & 4\sqrt{N-1} & 0 \\ 4\sqrt{N-1} & N-3 & N+1 \\ 0 & N+1 & N+1 \end{array}\right) & N\equiv 1\text{ mod }4
\end{array}.
$$
The remaining subspace squares to $\identity/4(N+1)$. Thus, the smallest absolute eigenvalue is $1/2\sqrt{N+1}$. Hence, $\sum G_n^2\sim N$, and the error dependence is $O(\epsilon^2N)$ in the worst case, but one anticipates that in typical cases, the dependence on $N$ is much weaker. 

\section{Analytic Solutions}

The disadvantage of the previous numerical technique is that the gap between eigenvalues scales as $1/N^2$, which means that $t_0\sim N^2$. This is much slower than we would like (after all, the longer we wait, the more noise is likely to build up), so it would be advantageous to find the fastest possible solutions. To that end, we provide some analytic solutions for spin chains with the best possible spectrum, which will yield $t_0\sim N$. These solutions don't permit us any control over the target eigenvector (except that different solutions have a different eigenvector), but by finding a solution that is as close as possible to the one that we want, we will be able to select from a number of perturbative techniques to drive the solution towards one that we want.
\begin{definition}
The $N\times N$ symmetric tridiagonal matrices with diagonal elements
$$
h_n=(N-1)\left(\frac{N+1}{2}+\alpha\right)-2\left(n-\frac{N+1}{2}\right)^2
$$
and off-diagonal elements
$$
K_n=\sqrt{n(n+\alpha)(N-n)(N+\alpha-n)}
$$
have a spectrum $k(k+2\alpha+1)$ for $k=0,\ldots, N-1$ \cite{albanese2004} and $\alpha\geq 0$. We call these the Hahn matrices.
\end{definition}
In fact, \cite{albanese2004} restricted the values of $\alpha$ more strongly, but this was because other specific properties of the spectrum were required. \cite{albanese2004} also gives the eigenvectors of these matrices in terms of the Hahn polynomials.

\begin{lemma} \label{lem:desham}
If we construct the $(2N+1)\times(2N+1)$ symmetric tridiagonal matrix with 0 on the main diagonal, and off-diagonal couplings that satisfy
$$
J_{2n-1}^2+J_{2n}^2=h_n+\left(\frac{2\alpha+1}{2}\right)^2\qquad J_{2n}J_{2n+1}=K_n,
$$
then this matrix has spectrum $0$ and $\pm\left\{\left(k+\frac{2\alpha-1}{2}\right)\right\}_{k=1}^N$, where we use the values from the $N\times N$ Hahn matrices of parameter $\alpha$.
\end{lemma}
In particular, we will be interested in integer values of $\alpha$ to create the spectrum that we desire ($t_0=2\pi$). This definition permits us to create a symmetric matrix by looking at the central coupling term $J_N$, since $J_N=J_{N+1}$. Hence,
$$
J_N=\left\{\begin{array}{cc}
\sqrt{\frac{h_{(N+1)/2}+(\alpha+\half)^2}{2}} & N\text{ odd}	\\
\sqrt{K_{N/2}} & N\text{ even.}
\end{array}\right. 
$$
\begin{proof}
Let $H_1$ be the matrix constructed in this way. Clearly, it anti-commutes with the operator $\sum_{n=1}^{2N+1}(-1)^n\proj{n}$, meaning that the eigenvalues arise in $\pm\lambda$ pairs, centred on a single 0 value (since the eigenvalues of such a matrix must be non-degenerate). So, let us consider $H^2_1$. Up to a permutation, this is equivalent to a block-diagonal matrix where one of the blocks is
$H_H+\left(\frac{2\alpha+1}{2}\right)^2\identity$, where $H_H$ is the $N\times N$ Hahn matrix. Thus, the spectrum of this block is
$$
k(k+2\alpha+1)+\left(\frac{2\alpha+1}{2}\right)^2=\left(k+\frac{2\alpha+1}{2}\right)^2
$$
$k=0,\ldots, N-1$. Since this block has eigenvalues $\left(k+\frac{2\alpha+1}{2}\right)^2$, $H$ must have eigenvalues with modulus $k+\frac{2\alpha+1}{2}$, and given that they arise in $\pm\lambda$ pairs, it must have both $\pm(k+\frac{2\alpha+1}{2})$ for $k=0,\ldots N-1$.
\end{proof}

For example, with $N=3$ and $\alpha=1$, we create the matrix
$$
H_1=\left(
\begin{array}{ccccccc}
 0 & \frac{7 \sqrt{\frac{3}{11}}}{2} & 0 & 0 & 0 & 0 & 0 \\
 \frac{7 \sqrt{\frac{3}{11}}}{2} & 0 & -4 \sqrt{\frac{2}{11}} & 0 & 0 & 0 & 0 \\
 0 & -4 \sqrt{\frac{2}{11}} & 0 & \frac{\sqrt{\frac{33}{2}}}{2} & 0 & 0 & 0 \\
 0 & 0 & \frac{\sqrt{\frac{33}{2}}}{2} & 0 & -\frac{\sqrt{\frac{33}{2}}}{2} & 0 & 0 \\
 0 & 0 & 0 & -\frac{\sqrt{\frac{33}{2}}}{2} & 0 & 4 \sqrt{\frac{2}{11}} & 0 \\
 0 & 0 & 0 & 0 & 4 \sqrt{\frac{2}{11}} & 0 & -\frac{7 \sqrt{\frac{3}{11}}}{2} \\
 0 & 0 & 0 & 0 & 0 & -\frac{7 \sqrt{\frac{3}{11}}}{2} & 0 \\
\end{array}
\right).
$$
It can be verified that its spectrum is $\left\{0,\pm\frac32,\pm\frac52,\pm\frac72\right\}$, and its 0-eigenvector, up to normalisation, is approximately
$$
\ket{1}+\ket{7}+1.072(\ket{3}+\ket{5}).
$$

It is interesting to observe that for $\alpha=1$, the 0 eigenvector is very close to $\sum_{n=1}^{N+1}\ket{n}/\sqrt{N+1}$ (the vector that we used for the impossibility proof in Lemma \ref{lem:nosol}) -- numerically we have created matrices of (odd) size up to 10003, and the overlap, $F$, with that target eigenvector is always at least 0.999 (up to some signs which can be corrected by changing the signs of the couplings). Equally this means that the overlap with the $W$ state is approximately $1/\sqrt{2}$. Consequently, it can serve as a crude starting for numerical schemes -- by judiciously changing the signs of the coupling strengths we can guarantee an overlap with any target state of approximately $\left(\sum_n|\alpha_{2n-1}|\right)\sqrt{2}/\sqrt{N+1}$ which is never too small.

\section{Speed Limits}

For a given target state $\ket{\psi_T}$ in Problem \ref{prob:main}, how small can the synthesis time, $t_0$, be made? The choice of spectrum in the above analytic construction was motivated by the insight from perfect state transfer \cite{yung2006} that by compressing the spectrum as much as possible, one achieves the minimum state transfer time for a given maximum coupling strength. Here we prove that those insights carry forward to the different spectral conditions that we impose for the state generation task. The following proof technique represents an improvement over \cite{yung2006} for the case of odd length chains.
\begin{lemma}\label{lem:speed}
A state generation task satisfying the construction presented in Lemma \ref{lem:reduction} for a chain of length $2N+1$ requiring time $t_0$ has a maximum coupling strength
$$
J_{\max}\geq\frac{\pi}{2t_0}\sqrt{N^2-\half}
$$
if the Hamiltonian is symmetric (i.e.\ $B_n=B_{2N+2-n}$ and $J_n^2=J_{2N+1-n}^2$).
\end{lemma}
Note that our previous construction satisfies this for $\alpha=0$ and odd $N$. 
\begin{proof}
We remove the freedom of $\identity$ shifts on the Hamiltonian by fixing $B_{N+1}=0$. Having done this, we observe that the imposed symmetry of the Hamiltonian splits the matrix into anti-symmetric and symmetric subspaces with mutually interlacing eigenvalues $\{\mu_k\}_{k=1}^N$ and $\{\nu_k\}_{k=1}^{N+1}$ respectively ($\nu_k<\mu_k<\nu_{k+1}$). All eigenvalues must have an integer spacing, except for a spacing of $\half$ either side of one special eigenvalue. Let's assume this special eigenvalue is $\mu_{\tilde k}$. We have that
$$
4J_{\max}^2\geq 4J_N^2=\Tr(SH^2)=\sum\eta_k^2-\sum\mu_k^2.
$$
If we use the bounds $\eta_k\geq \eta_1+2(k-1)-\delta_{k>\tilde k}$ and $\mu_k\leq\eta_{k+1}-1+\half\delta_{k=\tilde k}$, then one readily derives
$$
4J_{\max}^2\geq\eta_1^2+(2N-1)\eta_1+2N^2-N-\frac{1}{4},
$$
which is the smallest possible ($N^2-\half$) for the choice $\eta_1=\half-N$.

One can follow a similar calculation under the assumption that the special eigenvalue is $\eta_{\tilde k}$. In that case, one would derive
$
4J_{\max}^2\geq N^2.
$
\end{proof}
Of course, even for a symmetric target eigenvector, it is not necessary that the Hamiltonian be symmetric, and the method of Lemma \ref{lem:reduction} is far from unique, so this proof has limited applicability. Variants of this proof can address different assumptions. For example, we can exchange the symmetry assumption for assuming that all the magnetic fields are equal (i.e.\ 0 up to identity shifts in the Hamiltonian). In this case, $\Tr(H^2)=\sum\lambda_n^2=2\sum_nJ_n^2$, and we relate $J_n=\eta_{n-1}J_{n-1}/\eta_{n+1}$. Given all the $\lambda_n$ are separated by at least $2\pi/t_0$, a similar inequality can be derived, which proves that any solution with $B_n=0$ and a spectrum $0,\pm 1,\pm3,\pm 5,\ldots$ is optimal for a solution of this type. When we try to relax the $B_n=0$ assumption, we run out of sufficient information to make the bound as tight as possible. Nevertheless, by fixing $|J_n|\leq J_{\max}$ for all $n$, one gets
\begin{equation}
\frac{J_{\max}t_0}{\pi}\sqrt{\frac{2(N-1)+\sum_{n=1}^{N}\left(\frac{|\eta_{n-1}|+|\eta_{n+1}|}{\eta_n}\right)^2}{N(N-1)(N-2)}}\geq\frac{1}{\sqrt{3}}		\label{eqn:nottight}
\end{equation}
thanks to the relation $\eta_{n-1}J_{n-1}+\eta_{n+1}J_n=-\eta_nB_n$, and assuming $\eta_n\neq 0$.

More generally, \cite{bravyi2006} conveys that to generate a non-trivial correlation function between two regions separated by a distance $L$ requires at least a time $\sim L$ for a fixed maximum coupling strength. For instance, if we consider the two operators $O_A=Z_1$ and $O_B=Z_N$, and evaluate
$$
\sigma(\ket{\psi})=\bra{\psi}O_AO_B\ket{\psi}-\bra{\psi}O_A\ket{\psi}\bra{\psi}O_B\ket{\psi}
$$ 
then at the start of the evolution, wherever the excitation is initially localised, we have $\sigma(\ket{n})=0$, while the final state has $\sigma(\ket{\alpha})=-4\alpha_1^2\alpha_N^2$. Provided $\alpha_1\alpha_N$ is not exponentially small, \cite{bravyi2006} conveys that $J_{\max}t_0\sim \Omega(N)$, so the scaling relation is certainly optimal, even without the additional assumptions in Lemma \ref{lem:speed}.

\begin{figure}
\begin{center}
\includegraphics[width=0.45\textwidth]{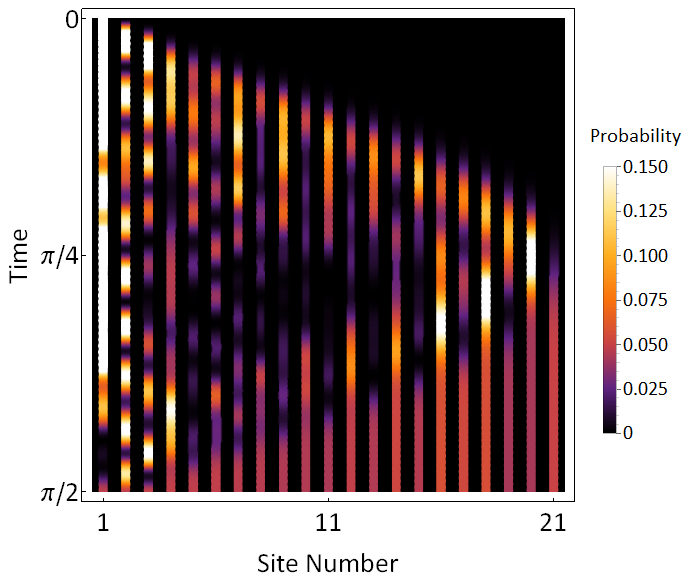}
\end{center}
\vspace{-0.5cm}
\caption{Time evolution of probability of excitation being at each site in a 21 qubit chain, approximating the evolution of $\ket{1}$ evolving to the $W$-state.}\label{fig:Wstate}
\vspace{-0.5cm}
\end{figure}

\section{Perturbative manipulations} The matrices that we have introduced in Lemma \ref{lem:desham} might have the ideal spectrum but each has a fixed 0-vector. If we want a different vector, we must apply an isospectral transformation. The following method has proven successful for systems of a few tens of qubits. We start with a Hamiltonian $H_1$ (couplings $J_n$ and fields $B_n$) and aim to make a new Hamiltonian which has the same spectrum, and whose 0-eigenvector is a better approximation to $\sum_n\eta_n\ket{n}$. The first step is to change the signs of the couplings (which makes no difference to the spectrum), to
$$
\text{sign}(J_n):=-\text{sign}\left(\frac{\eta_{n-1}J_{n-1}}{\eta_{n+1}}\right).
$$
because this minimises the norm of
$$
V=-\sum_{n=1}^{N}\frac{\eta_{n-1}J_{n-1}+\eta_{n+1}J_n}{\eta_n}\proj{n},
$$
making it as close to a perturbation as possible. We then follow an iterative procedure whereby we take $H_1+\delta V$, with $\delta=\min(1,\epsilon/\|V\|)$ for some $\epsilon\ll1$, calculate the eigenvectors $\ket{\tilde\lambda_n}$, and then find a new Hamiltonian (by following the Lanczos algorithm) using the target spectrum and the elements $\{\braket{1}{\tilde\lambda_n}\}$. The overall step is isospectral by construction, and should provide a small ($O(\epsilon)$) improvement in the accuracy of the target eigenvector. Thus, repetition is anticipated to drive us towards a good solution, should one exist. For example, Fig.\ \ref{fig:Wstate} depicts the evolution of a 21 qubit system which performs the evolution $\ket{1}\rightarrow\ket{\psi}$ where $\frac{1}{\sqrt{21}}\sum_{n=1}^{21}\braket{n}{\psi}\approx 1-2\times 10^{-15}$. With regards to the optimal speed, this example gives that $J_{\max}t_0=4.66$ while Eq.\ (\ref{eqn:nottight}) specifies that $J_{\max}t_0\geq 4.45$; there is little margin for finding a faster solution.


\section{Conclusions} We have shown that a spin chain can be engineered to create almost any single excitation state from its time evolution (up to local phases) vastly extending their utility. Our results can readily be applied to local free-fermion models (such as the transverse Ising model), or any one-dimensional nearest-neighbour Hamiltonian that is excitation preserving (such as the Heisenberg model). 

Any target state with no consecutive zero amplitudes can be realised. To get consecutive zeros, one could examine the technique that \cite{gladwell1986} specifies for fixing two eigenvectors of a matrix. While this gives no control over the spectrum, the procedure of Lemma \ref{lem:numeric} can be applied to get high accuracy solution. However, this can give no more than two consecutive zeros \footnote{For two eigenvectors \unexpanded{$|\eta^1\rangle$} and \unexpanded{$|\eta^2\rangle$} necessary conditions on there being a corresponding tridiagonal matrix include that $s_n=t_n=0$ or $s_nt_n>0$ for each $n=1,\ldots,N$, where \unexpanded{$s_n=\sum_{m=1}^n\eta_m^1\eta_m^2$} and $t_n=\eta_{n}^1\eta_{n+1}^2-\eta_{n+1}^1\eta_n^2$. However, the condition of two consecutive zeros is $\eta_n^1+\eta_n^2=0$ and $\eta_{n+1}^1+\eta_{n+1}^2=0$, which in turn means $t_n=0$, requiring $\eta_n^1=\eta_n^2=0$ such that $s_n=0$. This allows two zeros together, but to add a third consecutive zero would require two consecutive zeros in both eigenvectors.}. The challenge is to design systems that produce states with many 0 amplitudes, which is likely to require inordinate control over most of the eigenvectors. This is addressed in \cite{kay2016-b}.

{\em Acknowledgements:} We would like to thank L.\ Banchi and G.\ Coutinho for introductory conversations. This work was supported by EPSRC grant EP/N035097/1.

\bibliography{../../../References}

\end{document}